\documentclass[11pt]{article}

\usepackage{amsmath,amstext,amsthm,amssymb}

\newtheorem{theorem}{Theorem}
\newtheorem{corollary}[theorem]{Corollary}
\newtheorem{proposition}[theorem]{Proposition}
\theoremstyle{remark}

\renewcommand{\u}[1]{\underline{#1}}
\date{}

\title{An improvement of the upper bound for GKS communication game}
\author{Ivan Petrenko\thanks{The article
  was in part funded by RFBR according to the research project № 19-01-00563.}\\
  Moscow State University}

\begin{document}
\maketitle

\begin{abstract}
The GKS game was formulated by Justin Gilmer, Michal Kouck\'y, and Michael Saks in their research of the sensitivity conjecture. Mario Szegedy invented a protocol for the game with the cost of $O(n^{0.4732})$. Then a protocol with the cost of  $O(n^{0.4696})$ was obtained by DeVon Ingram who used a bipartite matching. We propose a slight improvement of Ingram's method and design a protocol with cost of $O(n^{0.4693})$.
\end{abstract}

\section{Introduction}

The GKS game is a three-player game where two players, Alice and Bob, play as a team  against their opponent Merlin. The game has two natural parameters $n,k$. In the first phase of the game Merlin and Alice fill an initially empty string $x$ of length $n$ with bits in the following fashion: First Merlin arbitrarily picks a permutation of the set $\{1,2,\dots,n\}$. Then he sends indices $1,2,\dots, n$  one-by-one in the chosen order to Alice. Once Alice receives an index $i$, she immediately has to pick a bit to fill the $i$th  position in the string, that is she chooses $x_i$. This goes on until only one empty position is left,  that position is filled by Merlin, not Alice. We will commonly refer to this bit as the last bit. The now-filled string is sent to Bob, whose goal is to determine the position filled by Merlin. To this end Bob chooses a set $S$ of $k$ natural numbers. Alice and Bob win if $S$ contains the position filled by Merlin, otherwise Merlin wins. Alice and Bob can communicate before the game starts but cannot communicate during the play.

A winning strategy for Alice and Bob in the $(k,n)$-game is called \emph{a $(k,n)$-strategy}. The minimal  $k$ for which there is a  $(k,n)$-strategy is called \emph{the cost of the $n$-game}.

The following proposition is very important for the study of the strategies for the GKS game:

\begin{proposition}\label{p1}
If there are a $(k, n)$-strategy and a $(k', n')$-strategy, then there exists a $(kk', nn')$-strategy.
\end{proposition}

The proof of this proposition can be found in [2]. This result allows to use $(k,n)$-strategies constructed for small values of $n$ to obtain upper bounds for the cost of the $n$-game for arbitrary $n$. More precisely, if a $(K,N)$-strategy exists, then the cost of the $n$-game is $O(n^{\log_NK})$, where the constant hidden in O-notation depends on $K,N$. Szegedy proved a $O(n^{0.4732})$ upper bound using $(5,30)$-strategy [2]. Ingram's $O(n^{0.4696})$ bound is derived from a $(11,165)$-strategy [3].

\section{Results}
\begin{theorem}\label{t2}
There exists a $(9,108)$-strategy.
\end{theorem}
\begin{proof}
The proof relies heavily on a table that is presented in Appendix A and was found by a computer search. This table consists of 220 binary strings of length 12. We call those strings \emph{codewords}. In every codeword three positions are underlined, all those positions contain ones. An important property of the table is that every possible combination of three underlined positions out of $\binom{12} 3 = 220$ combinations appears exactly once. Another property that is important for us is the following:

\begin{proposition}\label{p3}
Let x and y be different codewords. Let $x'$ and $y'$ be strings obtained from $x$ and $y$, respectively, by flipping at most one non-underlined bit. Then $x' \ne y'$.
\end{proposition}

Informally this property can be understood as follows: the code corrects one error provided it appears in a non-underlined position.
This is where our construction differs from that proposed by Ingram [3], who used Hamming codes that correct one  error (in any position).
Unfortunately, it is rather difficult to verify this property by hand and we cannot provide any proof of this fact.
We verified this property by a computer program. The table itself was also found by that program.

We will now construct a $(9,108)$-strategy.
First, Alice divides the string of length 108 into 9 blocks of length 12.
Alice then fills each block independently from other blocks as follows:

1.  The three first received positions in each  block are filled by  ones.
        
2. Once those three positions are filled, find the (unique) codeword in which  
those three positions are underlined and assign this codeword to the block.
        
3. The remaining positions except the last one are filled as in the assigned codeword.
        
4. The last position is filled so that the resulting block \emph{differs} from the assigned codeword.

Observe that every block completed by Alice alone is a codeword with one error in a non-underlined position.
From Proposition 3 we can conclude that every such block is different  from each codeword.
So Bob can receive two kinds of strings:

1. Either the received string has a block which is a codeword.         
In that case Bob knows that this block contains the last filled position and that that position is not underlined.
Thus, he forms the set of 9 non-underlined positions in that block.
        
2. Or, there is no such block in the received string,
that is, all the blocks are codewords with one  error, and that error occurs
in a non-underlined position. Then, due to Proposition~\ref{p3}, for each block Bob can find the assigned codeword and
the position of error, that position was filled last. The set of those positions has size 9 and contains the last bit.

In both cases Bob can find a set of 9 positions containing the last bit.
\end{proof}

\begin{corollary}
The cost of a GKS game is less then $O(n^{\log_{108}9})=O(n^{0.4693})$
\end{corollary}
This bound is derived by applying Proposition~\ref{p1} to Theorem~\ref{t2}.

\appendix

\section{Appendix}

\begin{tabular}{|r|r|r|}
\hline

    1. \u1\u1\u1000100000
    &
    2. \u1\u11\u110000000
    &
    3. \u1\u111\u11011010
         \\
    4. \u1\u1010\u1010110
    &
    5. \u1\u10010\u100001
    &
    6. \u1\u100001\u11000
         \\
    7. \u1\u1001000\u1011
    &
    8. \u1\u10000000\u100
    &
    9. \u1\u101000010\u10
         \\
    10. \u1\u1000110001\u1
    &
    11. \u10\u1\u110100010
    &
    12. \u11\u10\u10111001
         \\
    13. \u10\u100\u1101100
    &
    14. \u11\u1001\u110100
    &
    15. \u10\u10001\u10110
         \\
    16. \u10\u100000\u1000
    &
    17. \u10\u1010000\u110
    &
    18. \u11\u10000100\u10
         \\
    19. \u10\u100000010\u1
    &
    20. \u110\u1\u10110010
    &
    21. \u100\u10\u1100111
         \\
    22. \u100\u101\u101000
    &
    23. \u110\u1000\u11001
    &
    24. \u110\u11110\u1010
         \\
    25. \u110\u100111\u100
    &
    26. \u111\u1011010\u11
    &
    27. \u100\u10000000\u1
         \\
    28. \u1100\u1\u1101101
    &
    29. \u1011\u10\u111101
    &
    30. \u1100\u110\u10000
         \\
    31. \u1000\u1000\u1100
    &
    32. \u1001\u10010\u101
    &
    33. \u1100\u100101\u10
         \\
    34. \u1000\u1000001\u1
    &
    35. \u10001\u1\u100000
    &
    36. \u10000\u10\u10011
         \\
    37. \u10110\u101\u1010
    &
    38. \u10000\u1000\u110
    &
    39. \u10101\u10100\u10
         \\
    40. \u11100\u100011\u1
    &
    41. \u110100\u1\u10111
    &
    42. \u100110\u11\u1000
         \\
    43. \u100100\u100\u100
    &
    44. \u100000\u1000\u10
    &
    45. \u101001\u10000\u1
         \\
    46. \u1000000\u1\u1010
    &
    47. \u1000111\u10\u101
    &
    48. \u1000001\u111\u11
         \\
    49. \u1011011\u1010\u1
    &
    50. \u10110000\u1\u111
    &
    51. \u10100100\u10\u11
         \\
    52. \u10001101\u100\u1
    &
    53. \u100100010\u1\u10
    &
    54. \u100101001\u10\u1
         \\
    55. \u1010110111\u1\u1
    &
    56. 0\u1\u1\u110101010
    &
    57. 0\u1\u10\u10100011
         \\
    58. 0\u1\u100\u1000010
    &
    59. 0\u1\u1100\u100001
    &
    60. 0\u1\u11000\u10101
         \\
    61. 0\u1\u100000\u1100
    &
    62. 0\u1\u1111010\u111
    &
    63. 0\u1\u10000001\u11
         \\
    64. 0\u1\u100011100\u1
    &
    65. 0\u10\u1\u10010100
    &
    66. 1\u10\u10\u1000000
         \\
    67. 1\u10\u111\u110001
    &
    68. 0\u10\u1001\u10000
    &
    69. 0\u11\u10001\u1000
         \\
    70. 0\u10\u101000\u101
    &
    71. 0\u10\u1011000\u10
    &
    72. 0\u10\u10000100\u1
         \\
    73. 0\u100\u1\u1001000
    &
    74. 1\u110\u11\u110010
    &
    75. 0\u110\u100\u11010
         \\
    76. 1\u110\u1110\u1000
    &
    77. 1\u101\u11000\u100
    &
    78. 0\u100\u111010\u11
         \\
    79. 1\u110\u1001001\u1
    &
    80. 0\u1000\u1\u110110
    &
    81. 0\u1000\u10\u10000
         \\
    82. 0\u1011\u101\u1001
    &
    83. 0\u1101\u1001\u110
    &
    84. 0\u1001\u10100\u11
         \\
    85. 1\u1011\u101111\u1
    &
    86. 0\u10011\u1\u11000
    &
    87. 0\u10000\u10\u1000
         \\
    88. 0\u10110\u101\u100
    &
    89. 0\u10111\u1001\u11
    &
    90. 0\u10001\u11000\u1
         \\
    91. 1\u110000\u1\u1111
    &
    92. 0\u101010\u11\u100
    &
    93. 0\u100000\u101\u10
         \\
    94. 1\u111010\u1000\u1
    &
    95. 1\u1000100\u1\u110
    &
    96. 1\u1100010\u10\u10
         \\
    97. 0\u1100110\u111\u1
    &
    98. 0\u11100100\u1\u10
    &
    99. 0\u10000100\u10\u1
         \\
    100. 0\u100000110\u1\u1
    &
    101. 11\u1\u1\u10010110
    &
    102. 10\u1\u11\u1001001
         \\
    103. 00\u1\u111\u100011
    &
    104. 10\u1\u1000\u10000
    &
    105. 00\u1\u10000\u1010
         \\
    106. 01\u1\u111100\u100
    &
    107. 10\u1\u1000000\u11
    &
    108. 01\u1\u10100001\u1
         \\
    109. 01\u11\u1\u1010000
    &
    110. 00\u10\u10\u101000
    &
    111. 10\u11\u110\u10100
         \\
      \hline
\end{tabular}

\begin{tabular}{|r|r|r|}
\hline
    112. 01\u11\u1100\u1011
    &
    113. 01\u10\u11000\u101
    &
    114. 00\u10\u100000\u10
         \\
    115. 00\u10\u1111000\u1
    &
    116. 00\u100\u1\u100000
    &
    117. 00\u110\u11\u10100
         \\
    118. 00\u100\u100\u1001
    &
    119. 01\u111\u1111\u101
    &
    120. 00\u110\u10100\u10
         \\
    121. 00\u101\u100001\u1
    &
    122. 10\u1100\u1\u10011
    &
    123. 00\u1100\u10\u1100
         \\
    124. 00\u1011\u101\u101
    &
    125. 00\u1000\u1010\u11
    &
    126. 00\u1100\u11111\u1
         \\
    127. 11\u10010\u1\u1000
    &
    128. 01\u10010\u10\u100
    &
    129. 00\u10000\u100\u11
         \\
    130. 00\u10001\u1010\u1
    &
    131. 10\u110001\u1\u100
    &
    132. 00\u100001\u11\u10
         \\
    133. 00\u101001\u110\u1
    &
    134. 00\u1001100\u1\u10
    &
    135. 11\u1100101\u10\u1
         \\
    136. 01\u10011100\u1\u1
    &
    137. 000\u1\u1\u1000110
    &
    138. 000\u1\u10\u101010
         \\
    139. 000\u1\u100\u11000
    &
    140. 000\u1\u1111\u1111
    &
    141. 001\u1\u10000\u100
         \\
    142. 010\u1\u100000\u10
    &
    143. 000\u1\u1011001\u1
    &
    144. 001\u10\u1\u101010
         \\
    145. 001\u10\u10\u11001
    &
    146. 000\u10\u100\u1000
    &
    147. 101\u10\u1000\u100
         \\
    148. 110\u11\u10000\u11
    &
    149. 000\u11\u100000\u1
    &
    150. 000\u100\u1\u11110
         \\
    151. 000\u100\u11\u1001
    &
    152. 000\u111\u100\u101
    &
    153. 100\u101\u1100\u10
         \\
    154. 001\u110\u10100\u1
    &
    155. 100\u1100\u1\u1011
    &
    156. 000\u1000\u10\u101
         \\
    157. 000\u1000\u100\u10
    &
    158. 001\u1100\u1000\u1
    &
    159. 000\u11000\u1\u101
         \\
    160. 100\u10010\u10\u11
    &
    161. 110\u11010\u100\u1
    &
    162. 100\u110001\u1\u10
         \\
    163. 111\u110000\u10\u1
    &
    164. 101\u1110100\u1\u1
    &
    165. 0100\u1\u1\u100100
         \\
    166. 0001\u1\u11\u10000
    &
    167. 1000\u1\u100\u1010
    &
    168. 0000\u1\u1010\u100
         \\
    169. 1010\u1\u11001\u11
    &
    170. 0000\u1\u100111\u1
    &
    171. 1010\u10\u1\u10000
         \\
    172. 0000\u10\u11\u1100
    &
    173. 1010\u10\u100\u101
    &
    174. 0000\u10\u1001\u10
         \\
    175. 0000\u10\u10000\u1
    &
    176. 0010\u110\u1\u1000
    &
    177. 0110\u101\u10\u100
         \\
    178. 0100\u101\u100\u10
    &
    179. 0100\u100\u1000\u1
    &
    180. 1001\u1110\u1\u100
         \\
    181. 0010\u1111\u10\u11
    &
    182. 0110\u1000\u100\u1
    &
    183. 1100\u10100\u1\u11
         \\
    184. 0111\u10101\u11\u1
    &
    185. 0011\u100001\u1\u1
    &
    186. 01110\u1\u1\u11000
         \\
    187. 00001\u1\u10\u1001
    &
    188. 00000\u1\u101\u100
    &
    189. 00000\u1\u1110\u10
         \\
    190. 00000\u1\u10001\u1
    &
    191. 10001\u11\u1\u1110
    &
    192. 10000\u10\u11\u100
         \\
    193. 11110\u11\u111\u10
    &
    194. 00000\u11\u1110\u1
    &
    195. 00111\u100\u1\u100
         \\
    196. 00011\u101\u10\u10
    &
    197. 11000\u100\u100\u1
    &
    198. 01000\u1011\u1\u11
         \\
    199. 00000\u1000\u10\u1
    &
    200. 00010\u10101\u1\u1
    &
    201. 101001\u1\u1\u1001
         \\
    202. 000000\u1\u10\u100
    &
    203. 010101\u1\u110\u11
    &
    204. 100000\u1\u1000\u1
         \\
    205. 010101\u10\u1\u101
    &
    206. 101010\u10\u11\u10
    &
    207. 100000\u10\u110\u1
         \\
    208. 010000\u101\u1\u10
    &
    209. 010110\u110\u10\u1
    &
    210. 000100\u1001\u1\u1
         \\
    211. 1110100\u1\u1\u100
    &
    212. 0100101\u1\u11\u11
    &
    213. 0000000\u1\u100\u1
         \\
    214. 0011101\u10\u1\u10
    &
    215. 1100000\u10\u10\u1
    &
    216. 0000100\u101\u1\u1
         \\
    217. 01011001\u1\u1\u10
    &
    218. 00000000\u1\u11\u1
    &
    219. 00010100\u10\u1\u1
         \\
    220. 110100000\u1\u1\u1
    &&\\
      \hline
\end{tabular}
\end{document}